\documentclass{fundam}

\usepackage{graphicx}
\usepackage{tikz}
\usepackage{mymacros}
\usetikzlibrary{automata, positioning, arrows}
\usepackage{cite}
\usepackage{xspace}
\usepackage{amssymb}

\begin{document}


\setcounter{page}{241}
\publyear{22}
\papernumber{2160}
\volume{189}
\issue{3-4}

 \finalVersionForARXIV


\title{Reachability in Simple Neural Networks}

\author{Marco S\"alzer\thanks{Address for correspondence: School of Electr. Eng. and Computer Science,
		                 University of Kassel, Germany},  Martin Lange
 \\
		School of Electr. Eng. and Computer Science\\
		University of Kassel, Germany\\
       \{marco.saelzer, martin.lange\}@uni-kassel.de
}

\maketitle
\runninghead{M. S\"alzer and  M. Lange}{Reach. Is NP-Complete Even for the Simplest NN}

\begin{abstract} 	
We investigate the complexity of the reachability problem for (deep) neural networks: does it compute valid output given some
valid input? It was recently claimed that the problem is NP-complete for general neural networks and specifications over the
input/output dimension given by conjunctions of linear inequalities. We recapitulate the proof and repair some flaws in the original upper
and lower bound proofs. Motivated by the general result, we show that NP-hardness already holds for restricted classes of simple specifications and neural networks.
Allowing for a single hidden layer and an output dimension of one as well as neural networks with just one negative, zero and one positive weight or bias is sufficient to
ensure NP-hardness. Additionally, we give a thorough discussion and outlook of possible extensions for this direction of research on neural network verification.
\end{abstract}

\begin{keywords}
	machine learning, computational complexity, formal specification and verification
\end{keywords}


\section{Introduction}
Deep learning has proved to be very successful for highly challenging or even otherwise intractable tasks in a broad range of applications
such as image recognition \cite{KrizhevskySH17} or natural language processing \cite{X12a} but also safety-critical applications like
autonomous driving \cite{GrigorescuTCM20}, medical applications \cite{LitjensKBSCGLGS17}, or financial matters \cite{DixonKB17}. These
naturally come with safety concerns and the need for certification methods. Recent such methods can be divided into
(\textsc{i}) Adversarial Attack and Defense, (\textsc{ii}) Testing, and (\textsc{iii}) Formal Verification. A comprehensive survery about all three categories is given in \cite{HuangKRSSTWY20}.

The former two cannot guarantee the absence of errors. Formal verification of neural networks (NN) is a relatively new area of research which
ensures completeness of the certification procedure. Recent work on sound and complete verification algorithms for NN
is mostly concerned with efficient solutions to their reachability problem \reach{}
\cite{KatzBDJK17, Ehlers17, NarodytskaKRSW18, BunelTTKM18}:
given an NN and symbolic specifications of valid inputs and outputs, decide whether there is some valid input such that the
corresponding output is valid, too. This corresponds to the understanding of reachability in classical software verification: valid sets of inputs and outputs are specified and the question is whether there is a
valid input that leads to a valid output. Put differently, the question is whether the set of valid outputs is reachable from the set of valid
inputs. The difference to classical reachability problems in discrete state-based programs is that there reachability is a matter of
\emph{lengths} of a connection. In NN this is given by the number of layers, and it is rather the \emph{width} of the continuous
state space which may cause unreachability.

Solving \reach{} is interesting for practical purposes. An efficient algorithm can be used to ensure that no input from some specified
set of inputs is misclassified or that some undesired class of outputs is never reached. In applications like autonomous-driving, where
classifiers based on neural networks are used to make critical decisions, such safeguards are indispensable.

However, all known algorithms for \reach{} show the same drawback: a lack of scalability to networks of large size which, unfortunately,
are typically featured in real-world scenarios. This is not a big surprise as the problem is NP-complete. This
result was proposed by Katz et al.\ \cite{KatzBDJK17} for NN with ReLU and identity activations, and later also by Ruan et al.\
mentioned in \cite{RuanHK18} and done in \cite{RuanHK18_arxiv_version}.
While there is no reason to doubt the NP-completeness claim, the proofs are not stringent and contain flaws.

The argument for the upper bound in \cite{KatzBDJK17} claims inclusion in NP via the standard guess-and-verify approach: guess some
valid inputs, pass them through the NN and check whether the resulting numbers are valid outputs. This is flawed, though: it misses the
fact that the guessed witnesses need to be polynomially bounded in size, i.e.\ in the size of the NN and the input and
output specifications. No argument is given in \cite{KatzBDJK17} for a bound on the representation of such values, let alone a polynomial
one.\footnote{A follow-up paper by Katz et al.\ contains an extended version of the original article \cite{Katz21} without this flaw
being corrected.}
 In fact, guessing values in $\Real$ is not even effective without a bound on the size of their representation, hence this approach
does not even show recursive enumerability of the reachability problem. On the other hand, obtaining a polynomial bound on the values to be
guessed is closely linked to the question whether such values can be approximated up to some precision, for which no argument is given
in \cite{KatzBDJK17} either.

The arguments for the lower bound by a reduction from \threesat{} in \cite{KatzBDJK17} and \cite{RuanHK18_arxiv_version} rely on a
discretisation of real values to model Boolean values. This does not work for the signum function $\sigma$ used by Ruan et al.\
as it is not congruent for sums: e.g.\ $\sigma(-3) = \sigma(-1)$ but $\sigma(2 + (-3)) \ne \sigma(2 + (-1))$, showing that one
cannot simply interpret any negative number as the Boolean value \emph{false} etc. As a consequence, completeness of the
construction fails as there are (real) solutions to \reach{} which do not correspond to (discrete) satisfying \threesat assignments.
Katz et al.\ seem to be aware of this in \cite{KatzBDJK17} and use a slightly more elaborate discretisation in their reduction, but unfortunately it
still suffers from similar problems. These problems are repaired in \cite{Katz21}.

We start our investigations into the complexity of \reach{} by fixing these issues in Section~\ref{sec:np_compl}. We provide a different
argument for membership in NP which shows that the need for nondeterminism is not to be sought in the input values but in the use
of nodes with sophisticated activation functions like ReLU. Somehow surprisingly, we obtain the missing polynomial bound for the witnesses used
by Katz et al.\ as a corollary of our observations. Moreover, we show that NP membership is preserved if we allow the use of arbitrary piecewise linear activation functions.
We also address the issue of discretisation of real values in the lower bound proof, fixing the construction given by Katz et al.\ in \cite{KatzBDJK17}
and presenting an alternative fix compared to \cite{Katz21}.  We do not address the one by Ruan et al.\ from \cite{RuanHK18} any further, as this does not
provide further insights or new results.

We then observe that the reduction from \threesat constructs a very specific class of neural networks. NN from this class have a fixed number of
layers but scaling input and output dimension as well as layer size. This raises the question whether, in comparison to these networks,
reducing the number of layers or fixing dimensionality leads to a class of networks for which \reach{} is
efficiently solvable. In Section~\ref{sec:low_compl_fragments} we show that the answer to this is mostly negative: NP-hardness of \reach{}
holds for NN with just one layer and an output dimension of one. While this provides minimal requirements on the structure of NN for
\reach{} to be NP-hard, we also give minimal criteria on the weights and biases in NN for NP-hardness to hold. Thus, the computational
difficulty of \reach{} in the sense of NP-completeness is quite robust. The requirements on the structure or parameters of an NN
that are needed for NP-hardness to occur are easily met in practical applications.

We conclude in Section~\ref{sec:conclusion} with references to possible future work which focuses on further research regarding
computational complexity and decidability issues on the reachability problems of NN or deep learning models in a broader sense.

This article is a revised and expanded version of \cite{SalzerL21}, with the following noteworthy differences:
proofs and explanations, especially in Section~\ref{sec:low_compl_fragments}, are enriched with helpful explanations and descriptions,
the NP membership of \reach is shown for neural networks with arbitrary, piecewise linear activation functions (see Section~\ref{sec:np_compl;subsec:member}),
we provide a direct proof for the polynomial bound missing in the arguments of \cite{KatzBDJK17} (see Corollary~\ref{cor:polybound}),
it is shown that  NP-hardness holds for neural networks using ReLU activations only (see Corollary~\ref{cor:reluonly}) and
a thorough outlook on possible further steps is provided (see Section~\ref{sec:conclusion}).


\section{Preliminaries}
\label{sec:prelim}

A \emph{neural network} (NN) $N$ can be seen as a layered graph, representing a function of type $\Real^n \rightarrow \Real^m$.
The first layer $l=0$ is called the input layer and consists of $n$ nodes. The $i$-th node computes the output $y_{0i} = x_i$ where $x_i$ is the $i$-th input to the overall network. Thus, the output of the input layer  $(y_{00}, \dotsc, y_{0(n-1)})$ is identical to the input of $N$.
A  layer $1 \leq l \leq L-2$ is called \emph{hidden} and consists of $k$ nodes. Note that $k$ must not be uniform across the hidden layers of $N$. Then, the $i$-th node of layer $l$ computes the output $y_{li} = \sigma_{li}(\sum_j c^{(l-1)}_{ji}y_{(l-1)j} + b_{li})$ where $j$ iterates
over the output dimensions of the previous layer, $c^{(l-1)}_{ji}$ are rational constants which are called \emph{weights}, $b_{li}$ is a rational constant which is called \emph{bias} and $\sigma_{li}$ is some (typically non-linear) function called \emph{activation}.
The outputs of all nodes of layer $l$ combined gives the output $(y_{l0}, \dotsc, y_{l(k-1)})$ of the hidden layer.
The last layer $l = L-1$ is called the \emph{output layer} and consists of $m$ nodes. The $i$-th node computes an output $y_{(L-1)i}$ in the same way as a node in a hidden layer. The output of the output layer $(y_{(L-1)0}, \dotsc, y_{(L-1)(m-1)})$ is considered as the output of the network $N$.
We denote the output of $N$ given $\boldsymbol{x}$ by $N(\boldsymbol{x})$.%

We consider neural networks using piecewise linear activations in this work.
A function $f: \Real \rightarrow \Real$ is called \emph{piecewise linear (PWL)} if there are $a_0,b_0 \in \Rat$ such that $f(x) = a_0x + b_0$ or there are $m \in \Nat, m \geq 2$ linear functions $a_i x + b_i$ with $a_i, b_i \in \Rat$ and $m-1$ breakpoints $t_j \in \Real$
such that
\begin{align*}
	f(x) = \begin{cases}	a_0x + b_0 & \text{if } x < t_0, \\
										 a_i x + b_i & \text{if } t_i \leq x < t_{i+1}  \\
										 a_{m-1} x + b_{m-1} & \text{if } t_{m-1} \leq x.
	  \end{cases}
\end{align*}
We denote the $i$-th linear part of $f$ by $f^i$. The \emph{ReLU function}, commonly used as an activation in all hidden layers, is defined as $x \mapsto \max(0,x)$ or $x \mapsto 0$ if $x < 0$ and $x \mapsto x$ otherwise.
The second definition makes clear that ReLU is a PWL function. Obviously, the same holds for the identity function.
In order to make our results about lower bounds as strong as possible, we focus on neural networks using ReLU (and identity as a shortcut) activations only. Nodes with ReLU (identity) activation are called ReLU (identity) nodes.
Given some input to the NN, we say that a ReLU node is \emph{active}, resp.\ \emph{inactive} if the input for its activation function is greater or equal, resp.\ less than  $0$.

\medskip
We make use of two ways to represent NN in this work. The first one is visualizing an NN as a directed graph with weighted edges. An example is given in Figure~\ref{sec:prelim;fig:nn_visualized}.
The second one is referring to a NN $N$ by its computed function $f_N$. Obviously, this second way is ambiguous as there are infinitely many NN that compute $f_N$. To keep things easy,
we assume that it is clear in the respective situation which NN is referred to.

\begin{figure}[ht!]
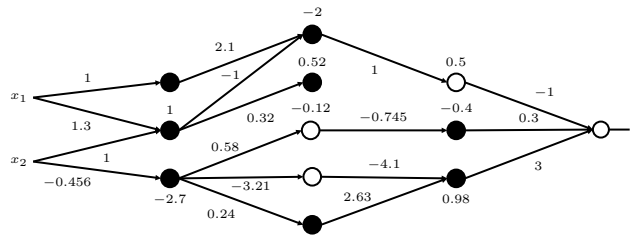

\vspace*{-1mm}
	\centering
	\include{graphics/prelim_network_schema}\vspace*{-5mm}
	\caption{Schema of a neural network with five layers, input dimension of two and output dimension of one. Filled nodes are
       ReLU nodes, empty nodes are identity nodes.  An edge between two nodes $u$ and $v$ with label $w$ denotes that the
       output of $u$ is weighted with $w$ in the computation of $v$. No edge between $u$ and $v$ implies $w=0$.
		The bias of a node is depicted by a value above or below the node. If there is no such value then the bias is zero.}
	\label{sec:prelim;fig:nn_visualized}\vspace*{-2mm}
\end{figure}

\medskip
Our main interest lies in the validity of specifications over the output values of NN given specifications over their input values.
A \emph{specification} $\varphi$ for a given set of variables $X$ is defined by the following grammar:
\begin{align*}
	\varphi &::= \varphi \land \varphi \mid t \leq b \qquad &
	t &::= c \cdot x \mid t + t
\end{align*}
where $b,c$ are rational constants and $x \in X$ is a variable.
We use $t \geq b$ and $t = b$ as syntactic sugar for $-t \leq -b$ and $t \leq b \land -t \leq -b$. Furthermore, we use $\top$ for $x + (-x) = 0$ and $\bot$ for $x + (-x) = 1$ where $x$ is some variable.
We call a specification $\varphi$ \textit{simple} if for all $t \leq b$ it holds that $t = c\cdot x$ for some rational constant $c$ and variable $x$.
A specification $\varphi(x_0, \dotsc, x_{n-1})$ is \emph{true} under $\boldsymbol{x} = (r_0, \dotsc, r_{n-1}) \in \Real^n$
if each inequality in $\varphi$ is satisfied in real arithmetic with each $x_i$ set to $r_i$.
We write $\varphi(\boldsymbol{x})$ for the application of $\boldsymbol{x}$ to the variables of $\varphi$. 
If there are less variables in $\varphi$ than dimensions in $\boldsymbol{x}$ we ignore the additional values of $\boldsymbol{x}$. If we
consider a specification $\varphi$ in context of a neural network $N$ we call it an \emph{input}, respectively \emph{output specification} and assume that
the set of variables occurring in $\varphi$ is a subset of the input respectively output variables of $N$.

\medskip
The decision problem \reach{} is defined as follows: given a neural network $N$ with PWL activations, input specification $\varphi_{\text{in}}(x_0, \dotsc, x_{n-1})$ and output specification
$\varphi_{\text{out}}(y_0, \dotsc, y_{m-1})$, is there $\boldsymbol{x} \in \Real^n$ such that $\varphi_{\text{in}}(\boldsymbol{x})$ and $\varphi_{\text{out}}(N(\boldsymbol{x}))$ are true?


\section{\reach is NP-complete}
\label{sec:np_compl}

\subsection{Membership in NP}
\label{sec:np_compl;subsec:member}
The argument used by Katz et al.\ to show membership of \reach{} in NP can be summarized as follows: nondeterministically guess an input
vector $\boldsymbol{x}$ as a witness, compute the output $N(\boldsymbol{x})$ of the network and check that
$\varphi_{\text{in}}(\boldsymbol{x}) \land \varphi_{\text{out}}(N(\boldsymbol{x}))$ holds. It is indisputable that the computation and
check of this procedure are polynomial in the size of $N$, $\varphi_{\text{in}}$, $\varphi_{\text{out}}$ \emph{and} the size of
$\boldsymbol{x}$. However, for inclusion in NP we also need the size of $\boldsymbol{x}$ to be polynomially bounded in the size of
the instance given as $(N,\varphi_{\text{in}},\varphi_{\text{out}})$, which is not provided in \cite{KatzBDJK17}.
We first present an alternative proof for $\reach \in \text{NP}$, circumventing the missing argument of a polynomial bound for the size of satisfying inputs,
by establishing a connection between NN with PWL activations and an extended notion of linear programs. Note that this subsumes the claim of Katz et al.,
as ReLU and identity are PWL activations. Furthermore, this proof directly implies the polynomial bound in question, which we will see at the end of this
section.

\begin{definition}
A \emph{PWL-linear program} over a set $X = \{x_0,\ldots,x_{n-1}\}$ of variables is a set $\Phi$ of (in-)equalities of the form
$b_j + \textstyle\sum_{i=1}^m c_{ji} \cdot x_{ji} \leq x_j$ and $f(b_j + \textstyle\sum_{i=1}^m c_{ji} \cdot x_{ji}) = x_j$
where $x_{ji}, x_j \in X$, $c_{ji},b_j \in \Rat$ and $f$ is a PWL function. We call the second kind a \emph{PWL-equality}.
A \emph{solution} to $\Phi$ is a vector $\boldsymbol{x} \in \Real^n$ which satisfies all (in-)equalities when each variable $x_i \in X$ is replaced by $\boldsymbol{x}(i)$.
A PWL-equality $f(b_j + \sum_{i=1}^m c_{ji} \cdot x_{ji}) = x_j$, where $f$ has $m$ linear parts $a_l x + b_l$ and $m-1$ breakpoints $t_l$,
is satisfied by $\boldsymbol{x}$ with  $t_l \leq b_j + \sum_{i=1}^m c_{ji} \cdot x_{ji} < t_{l+1}$ if  $a_l(b_j + \sum_{i=1}^m c_{ji} \cdot x_{ji}) + b_l = x_j$.
The problem of \emph{solving} a PWL-linear program is: given $\Phi$, decide whether there is a solution to it.
\end{definition}

It is well-known that regular linear programming can be done in polynomial time \cite{Karmarkar84}. We use this to show that solving
PWL-linear programs can be done in nondeterministic polynomial time.

\begin{lemma}
\label{thm:relulinprognp}
The problem of solving a PWL-linear program is in NP.
\end{lemma}
\begin{proof}
	Suppose a PWL-linear program $\Phi$ with $l$ PWL-equalities is given. First, note that there is a maximum number $k$ of linear parts a PWL function $f$ occurring in $\Phi$ consists of.
	Then, the existence of a solution can be decided as follows. Guess, for each PWL-equality $\chi_h$ of the form $f_h(b_j + \sum_{i=1}^m c_{ji} \cdot x_{ji}) = x_j$ where $f_h$ has $k' \leq k$ parts,
	some $w_h \in \{0, \dotsc, k'-1\}$, leading to the witness $\boldsymbol{w} = (w_0, \ldots, w_{l-1})$. It is not hard to see that its size is polynomially bounded by the size of $\Phi$.
	Next, let the linear program $\Phi_{\boldsymbol{w}}$ result from $\Phi$ by replacing each $\chi_h$ by the (in-)equalities $b_j + \textstyle\sum_{i=1}^m c_{ji} \cdot x_{ji} \ge t_{w_h}$,
	$b_j + \textstyle\sum_{i=1}^m c_{ji} \cdot x_{ji} \leq t_{w_h+1} + z_h$ and
	$a_{w_h}(b_j + \textstyle\sum_{i=1}^m c_{ji} \cdot x_{ji}) + b_{w_h} = x_j$ where $z_h$ is a fresh variable.

\medskip	
	The following is straightforward:
	(\textsc{i}) $\Phi_{\boldsymbol{w}}$ is linear in the size of $\Phi$.
	(\textsc{ii}) An optimal solution for the minimization of $\boldsymbol{z} = (z_1, \dotsc, z_l)$ over the constraints given by $\Phi_{\boldsymbol{w}}$
	where all $z_i$ are strictly negative is a witness for a solution of $\Phi$.
	(\textsc{iii}) If $\Phi$ has a solution, then there is $\boldsymbol{w} \in \{0, \dotsc k-1\}^l $ such that $\Phi_{\boldsymbol{w}}$
	has a solution. This can be created as follows. Let $\boldsymbol{x}$ be a solution to $\Phi$. For each PWL-equality $\chi_h$ as
	above, let $w_h = j$ if the active linear part is the $j$-th. Then $\boldsymbol{x}$ is also a solution for
	$\Phi_{\boldsymbol{w}}$. Thus, PWL-linear programs can be solved in nondeterministic polynomial time by guessing $\boldsymbol{w}$,  constructing
	the linear program $\Phi_{\boldsymbol{w}}$, finding the optimal solution minimizing $\boldsymbol{z}$ and checking if all $z_i$ are negative.
\end{proof}

Using the definition of a PWL-linear program and the corresponding lemma at hand, we are set to prove NP-membership of \reach.
\begin{theorem}
\label{thm:reachinnp}
\reach is in NP.
\end{theorem}

\begin{proof}
Let $\mathcal{I} = (N,\varphi_{\text{in}},\varphi_{\text{out}})$. We construct a PWL-linear program $\Phi_{\mathcal{I}}$ of size
linear in $|N|+|\varphi_{\text{in}}|+|\varphi_{\text{out}}|$ solvable if and only if there is a solution for
$\mathcal{I}$. The PWL-linear program $\Phi_{\mathcal{I}}$ contains the following (in-)equalities:
\begin{itemize}
\itemsep=0.9pt
\item $\varphi_{\text{in}}$ and $\varphi_{\text{out}}$ (with each conjunct seen as one (in-)equality) and
\item for each node $v_{li}$ computing the function $f(\sum_j c^{(l-1)}_{ji}y_{(l-1)j} + b_{li})$ it contains the PWL-equality
      $f(\sum_j c^{(l-1)}_{ji}y_{(l-1)j} + b_{li}) = y_{li}$.
\end{itemize}
The claim on the size of $\Phi_{\mathcal{I}}$ should be clear. Moreover, note that a solution $\boldsymbol{x}$ to $\mathcal{I}$ can be extended to an assignment $\boldsymbol{x}'$ of real values at every node of $N$, including values
$\boldsymbol{y}$ for the output nodes of $N$ s.t., in particular $N(\boldsymbol{x}) = \boldsymbol{y}$. Then $\boldsymbol{x}'$ is a solution
to $\Phi_{\mathcal{I}}$. Likewise, a solution to $\Phi_{\mathcal{I}}$ can be turned into a solution to $\mathcal{I}$ by projection
to the input variables.
Hence, \reach polynomially reduces to the problem of solving PWL-linear programs which, by Lemma~\ref{thm:relulinprognp} is in NP.
\end{proof}

Taking a look at our translation from \reach{} to PWL-linear programs, one could think that Theorem~\ref{thm:reachinnp} can also be proven by
translating \reach{} instances into \textit{mixed-integer linear programs} (MILP) \cite{AkitundeLMP18}.
Roughly speaking, an MILP is a LP with additional integer variables.
Deciding feasability of MILP is known to be NP-complete \cite{Karp72}.
Translating \reach{} instance into MILP works similarly as described in the proof of Theorem~\ref{thm:reachinnp} up to the point of PWL nodes.
For example, consider  a node $v_j$ computing the function $ReLU(b_j + \sum_{i=1}^m c_{ji} \cdot x_{ji})$.
First, we add the (in-)equalities $b_j + \sum_{i=1}^m c_{ji} \cdot x_{ji} = y_j-s_j, y_j \geq 0$ and $s_j \geq 0$ to our MILP.
Second, we add the inequalities $y_j \leq M^+_j(1-z_j)$ and $s_j \leq M^-_jz_j$ where $z_j$ is an integer variable and $M^+_j, M^-_j$ are sufficiently
large positive constants such that $-M^-_j \leq b_j + \sum_{i=1}^m c_{ji} \cdot x_{ji} \leq M^+_j$ holds.
This ensures that $z_j = 0$ corresponds to $v_j$ being active and $z_j=1$ corresponds to $v_j$ being inactive. For general PWL activation
functions the idea is the same, making use of switch variables like $z_j$ and bounds like $M^-, M^+$ in a more elaborate fashion.
However, this translation requires that suitable bounds exist for each node, which in general
is only guaranteed if each input dimension is bounded by the input specification.

\medskip
As stated in the beginning, our proof implies the polynomial bound missing in the claim of Katz et al.
We recall the missing argument in \cite{KatzBDJK17}: it is not obvious, that there is an input $\boldsymbol{x}$ for each positive \reach instance, witnessing the validity $I$,
which is polynomial bounded by the size of $I$.
Taking a look at the proof of Lemma~\ref{thm:relulinprognp}, we see that the solution of a PWL-linear program
$\Phi$ is determined by the solution of a regular linear program $\Phi_{\boldsymbol{w}}$ linear in the size of $\Phi$. Furthermore, from the definition of neural networks
and specifications follows that all coefficients occurring in $\Phi_{\boldsymbol{w}}$ are rational. Then, we can use
a standard result about linear programs: if a linear program $L$ with rational coefficients has an optimal solution then it has an optimal one of size
polynomially bounded by $L$ (for example, see Theorem~4.4 in \cite{KorteV06}). Thus, our proof of Lemma~\ref{thm:relulinprognp}
implies that each PWL-linear program, resulting
from the translation of a \reach instance, has a solution of polynomial size and, thus, using the arguments of the proof of Theorem~\ref{thm:reachinnp}, we retrieve the missing bound.
\begin{corollary}
	Let $I = (N, \varphi_{\text{in}}, \varphi_{\text{out}})$ be a \reach instance. If  $I$ is a positive instance of \reach, then there is an input $\boldsymbol{x}$ satisfying $\varphi_{\text{in}}$
	such that $N(\boldsymbol{x})$ satisfies $\varphi_{\text{out}}$ and the size of $\boldsymbol{x}$ is polynomially bounded by the size of $I$.
	\label{cor:polybound}
\end{corollary}

Additionally, it is interesting to point out the role of witnesses for positive instances of the \reach problem: the arguments leading to the result of Theorem~\ref{thm:reachinnp} above show
that an assignment to the PWL nodes determining their active linear part serves as a witness.
This immediately yields a polynomial fragment of \reach.
\begin{corollary}
\reach for NN with a bounded number of nodes with PWL activation is decidable in polynomial time.
\label{sec:np_compl;cor:relu_nodes}
\end{corollary}

\subsection{NP-hardness}

A natural candidate for a lower bound proof for \reach{} is a polynomial reduction from \threesat{}, as attempted by Katz et al.\
\cite{KatzBDJK17} and Ruan et al.\ \cite{RuanHK18}. The underlying idea is to encode the structure
of a \threesat{} formula in a neural network and the existence of a satisfying assignment for this formula in the corresponding
input- and output-specifications. To ease the definition of the NN resulting from this reduction Katz et al.\ introduced the notion of
gadgets.

\begin{figure}[ht!]
\vspace*{-1mm}
	\centering
	\tikzset{every picture/.style={line width=0.75pt}} 

\begin{tikzpicture}[x=0.75pt,y=0.75pt,yscale=-.8,xscale=.8]
	
	\draw   (151.39,76) .. controls (151.39,73.24) and (153.63,71) .. (156.39,71) .. controls (159.15,71) and (161.39,73.24) .. (161.39,76) .. controls (161.39,78.76) and (159.15,81) .. (156.39,81) .. controls (153.63,81) and (151.39,78.76) .. (151.39,76) -- cycle ;
	\draw    (91.89,76) -- (148.39,76) ;
	\draw [shift={(151.39,76)}, rotate = 180] [fill={rgb, 255:red, 0; green, 0; blue, 0 }  ][line width=0.08]  [draw opacity=0] (3.57,-1.72) -- (0,0) -- (3.57,1.72) -- cycle    ;
	\draw    (161.39,76) -- (221.39,76) ;
	\draw  [fill={rgb, 255:red, 0; green, 0; blue, 0 }  ,fill opacity=1 ] (316.39,79.5) .. controls (316.39,76.46) and (318.86,74) .. (321.89,74) .. controls (324.93,74) and (327.39,76.46) .. (327.39,79.5) .. controls (327.39,82.54) and (324.93,85) .. (321.89,85) .. controls (318.86,85) and (316.39,82.54) .. (316.39,79.5) -- cycle ;
	\draw    (276.39,45) -- (314.12,77.54) ;
	\draw [shift={(316.39,79.5)}, rotate = 220.78] [fill={rgb, 255:red, 0; green, 0; blue, 0 }  ][line width=0.08]  [draw opacity=0] (3.57,-1.72) -- (0,0) -- (3.57,1.72) -- cycle    ;
	\draw    (276.39,115) -- (313.15,82.49) ;
	\draw [shift={(315.39,80.5)}, rotate = 138.5] [fill={rgb, 255:red, 0; green, 0; blue, 0 }  ][line width=0.08]  [draw opacity=0] (3.57,-1.72) -- (0,0) -- (3.57,1.72) -- cycle    ;
	\draw    (266.39,79.5) -- (313.39,79.5) ;
	\draw [shift={(316.39,79.5)}, rotate = 180] [fill={rgb, 255:red, 0; green, 0; blue, 0 }  ][line width=0.08]  [draw opacity=0] (3.57,-1.72) -- (0,0) -- (3.57,1.72) -- cycle    ;
	\draw    (327.39,79.5) -- (353.54,79.5) ;
	\draw [shift={(356.54,79.5)}, rotate = 180] [fill={rgb, 255:red, 0; green, 0; blue, 0 }  ][line width=0.08]  [draw opacity=0] (3.57,-1.72) -- (0,0) -- (3.57,1.72) -- cycle    ;
	\draw   (356.54,79.5) .. controls (356.54,76.46) and (359,74) .. (362.04,74) .. controls (365.07,74) and (367.54,76.46) .. (367.54,79.5) .. controls (367.54,82.54) and (365.07,85) .. (362.04,85) .. controls (359,85) and (356.54,82.54) .. (356.54,79.5) -- cycle ;
	\draw    (367.54,79.5) -- (396.25,79.5) ;
	\draw  [dash pattern={on 4.5pt off 4.5pt}] (276.39,35) -- (386.39,35) -- (386.39,125) -- (276.39,125) -- cycle ;
	\draw    (266.39,45) -- (276.39,45) ;
	\draw    (266.39,115) -- (276.39,115) ;
	\draw  [dash pattern={on 4.5pt off 4.5pt}] (101.39,35) -- (211.39,35) -- (211.39,125) -- (101.39,125) -- cycle ;
	\draw  [fill={rgb, 255:red, 0; green, 0; blue, 0 }  ,fill opacity=1 ] (334.25,188) .. controls (334.25,184.96) and (336.71,182.5) .. (339.75,182.5) .. controls (342.79,182.5) and (345.25,184.96) .. (345.25,188) .. controls (345.25,191.04) and (342.79,193.5) .. (339.75,193.5) .. controls (336.71,193.5) and (334.25,191.04) .. (334.25,188) -- cycle ;
	\draw  [fill={rgb, 255:red, 0; green, 0; blue, 0 }  ,fill opacity=1 ] (334.25,228) .. controls (334.25,224.96) and (336.71,222.5) .. (339.75,222.5) .. controls (342.79,222.5) and (345.25,224.96) .. (345.25,228) .. controls (345.25,231.04) and (342.79,233.5) .. (339.75,233.5) .. controls (336.71,233.5) and (334.25,231.04) .. (334.25,228) -- cycle ;
	\draw    (290.75,208.5) -- (331.54,189.28) ;
	\draw [shift={(334.25,188)}, rotate = 154.77] [fill={rgb, 255:red, 0; green, 0; blue, 0 }  ][line width=0.08]  [draw opacity=0] (3.57,-1.72) -- (0,0) -- (3.57,1.72) -- cycle    ;
	\draw    (290.75,208.5) -- (331.51,226.77) ;
	\draw [shift={(334.25,228)}, rotate = 204.15] [fill={rgb, 255:red, 0; green, 0; blue, 0 }  ][line width=0.08]  [draw opacity=0] (3.57,-1.72) -- (0,0) -- (3.57,1.72) -- cycle    ;
	\draw   (374.25,208.5) .. controls (374.25,205.74) and (376.49,203.5) .. (379.25,203.5) .. controls (382.01,203.5) and (384.25,205.74) .. (384.25,208.5) .. controls (384.25,211.26) and (382.01,213.5) .. (379.25,213.5) .. controls (376.49,213.5) and (374.25,211.26) .. (374.25,208.5) -- cycle ;
	\draw    (345.25,188) -- (371.8,206.77) ;
	\draw [shift={(374.25,208.5)}, rotate = 215.26] [fill={rgb, 255:red, 0; green, 0; blue, 0 }  ][line width=0.08]  [draw opacity=0] (3.57,-1.72) -- (0,0) -- (3.57,1.72) -- cycle    ;
	\draw    (345.25,228) -- (371.76,210.17) ;
	\draw [shift={(374.25,208.5)}, rotate = 146.08] [fill={rgb, 255:red, 0; green, 0; blue, 0 }  ][line width=0.08]  [draw opacity=0] (3.57,-1.72) -- (0,0) -- (3.57,1.72) -- cycle    ;
	\draw    (384.25,208.5) -- (424.25,208.5) ;
	\draw  [dash pattern={on 4.5pt off 4.5pt}] (250.75,163.5) -- (414.17,163.5) -- (414.17,253.5) -- (250.75,253.5) -- cycle ;
	\draw   (496.75,80) .. controls (496.75,76.96) and (499.21,74.5) .. (502.25,74.5) .. controls (505.29,74.5) and (507.75,76.96) .. (507.75,80) .. controls (507.75,83.04) and (505.29,85.5) .. (502.25,85.5) .. controls (499.21,85.5) and (496.75,83.04) .. (496.75,80) -- cycle ;
	\draw    (457.25,55) -- (493.74,78.86) ;
	\draw [shift={(496.25,80.5)}, rotate = 213.18] [fill={rgb, 255:red, 0; green, 0; blue, 0 }  ][line width=0.08]  [draw opacity=0] (3.57,-1.72) -- (0,0) -- (3.57,1.72) -- cycle    ;
	\draw    (457.25,105) -- (493.71,82.1) ;
	\draw [shift={(496.25,80.5)}, rotate = 147.86] [fill={rgb, 255:red, 0; green, 0; blue, 0 }  ][line width=0.08]  [draw opacity=0] (3.57,-1.72) -- (0,0) -- (3.57,1.72) -- cycle    ;
	\draw    (507.75,80) -- (566,80) ;
	\draw    (436.25,55) -- (457.25,55) ;
	\draw    (436.25,105) -- (457.25,105) ;
	\draw  [dash pattern={on 4.5pt off 4.5pt}] (447.25,35) -- (557.25,35) -- (557.25,125) -- (447.25,125) -- cycle ;
	\draw    (239.5,208.5) -- (292.89,208.5) ;
	
	\draw (126.39,72.6) node [anchor=south] [inner sep=0.75pt]  [font=\tiny,xscale=0.8,yscale=0.8]  {$-1$};
	\draw (156.39,67.6) node [anchor=south] [inner sep=0.75pt]  [font=\tiny,xscale=0.8,yscale=0.8]  {$1$};
	\draw (154.89,33) node [anchor=south] [inner sep=0.75pt]  [font=\small,xscale=0.8,yscale=0.8] [align=left] {not};
	\draw (298.89,63.35) node [anchor=south west] [inner sep=0.75pt]  [font=\tiny,xscale=0.8,yscale=0.8]  {$-1$};
	\draw (291.39,76.1) node [anchor=south] [inner sep=0.75pt]  [font=\tiny,xscale=0.8,yscale=0.8]  {$-1$};
	\draw (297.89,101.15) node [anchor=north west][inner sep=0.75pt]  [font=\tiny,xscale=0.8,yscale=0.8]  {$-1$};
	\draw (321.89,88.4) node [anchor=north] [inner sep=0.75pt]  [font=\tiny,xscale=0.8,yscale=0.8]  {$1$};
	\draw (341.96,76.1) node [anchor=south] [inner sep=0.75pt]  [font=\tiny,xscale=0.8,yscale=0.8]  {$-1$};
	\draw (331.89,32.67) node [anchor=south] [inner sep=0.75pt]  [font=\small,xscale=0.8,yscale=0.8] [align=left] {or};
	\draw (362.04,70.6) node [anchor=south] [inner sep=0.75pt]  [font=\tiny,xscale=0.8,yscale=0.8]  {$1$};
	\draw (339.75,179.1) node [anchor=south] [inner sep=0.75pt]  [font=\tiny,xscale=0.8,yscale=0.8]  {$\epsilon $};
	\draw (339.75,236.9) node [anchor=north] [inner sep=0.75pt]  [font=\tiny,xscale=0.8,yscale=0.8]  {$\epsilon -1$};
	\draw (310.5,194.85) node [anchor=south east] [inner sep=0.75pt]  [font=\tiny,xscale=0.8,yscale=0.8]  {$-1$};
	\draw (310.5,221.65) node [anchor=north east] [inner sep=0.75pt]  [font=\tiny,xscale=0.8,yscale=0.8]  {$1$};
	\draw (361.75,194.85) node [anchor=south west] [inner sep=0.75pt]  [font=\tiny,xscale=0.8,yscale=0.8]  {$1$};
	\draw (361.75,221.65) node [anchor=north west][inner sep=0.75pt]  [font=\tiny,xscale=0.8,yscale=0.8]  {$1$};
	\draw (332.75,161.5) node [anchor=south] [inner sep=0.75pt]  [font=\small,xscale=0.8,yscale=0.8] [align=left] {bool$\displaystyle _{\epsilon }$};
	\draw (452.85,69) node [anchor=north west][inner sep=0.75pt]  [font=\small,rotate=-90,xscale=0.8,yscale=0.8]  {$\dotsc $};
	\draw (502.75,33) node [anchor=south] [inner sep=0.75pt]  [font=\small,xscale=0.8,yscale=0.8] [align=left] {and$\displaystyle _{n}$};
	\draw (478.75,64.35) node [anchor=south west] [inner sep=0.75pt]  [font=\tiny,xscale=0.8,yscale=0.8]  {$1$};
	\draw (478.75,96.15) node [anchor=north west][inner sep=0.75pt]  [font=\tiny,xscale=0.8,yscale=0.8]  {$1$};

\end{tikzpicture}\vspace*{-4mm}
	\caption{Gadgets used in the reduction from \threesat{} to \reach{}. A non-weighted outgoing edge of a gadget is connected to a weighted incoming edge of another gadget in the actual construction or is considered an output of the overall neural network.}
	\label{sec:np_compl;fig:3sat_gadgets}\vspace*{-3mm}
\end{figure}
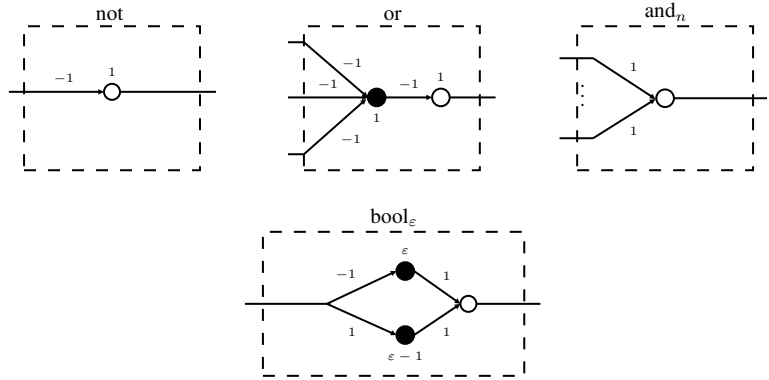

\begin{definition}
	Let $\epsilon \in \Real$ be a small constant and $n \in \Nat$. We call NN, which compute one of the following functions, \emph{gadgets}:
	\begin{align*}
		 \notg(x) &= \id(1-x) \\
		 \org(x_1, x_2, x_3) &= \id(1-\ReLU(1-x_1-x_2-x_3)) \\
		 \andg(x_1, \dotsc, x_n) &= \id(x_1 + \dotsb + x_n) \\
		 \boolg(x) &= \id(\ReLU(\epsilon-x) + \ReLU(x+\epsilon-1))
	\end{align*}
	\label{sec:np_compl;def:gadgets}
\end{definition}

The obvious candidates for these gadgets are visualized in Figure~\ref{sec:np_compl;fig:3sat_gadgets}.
The name of the first three gadgets gives a direct indication of their intended properties.
\begin{lemma}
	Let $n \in \Nat$ be a natural number and all $r_i \in \{0,1\}$. It holds that
	\begin{enumerate}
    \itemsep=0.8pt
		\item $\notg(1) = 0$ and $\notg(0)=1$,
		\item $\org(0,0,0) = 0$,
		\item \label{prop:org_feature} $\org(r_1, r_2, r_3) = 1$ if and only if at least one $r_i$ is equal to $1$,
		\item \label{prop:andg_feature} $\andg(r_1, \dotsc, r_n) = n$ if and only if all $r_i$ are equal to $1$.
	\end{enumerate}
	\label{sec:np_completeness;lem:gadget_props}
\end{lemma}	
\begin{proof}
	These are straightforward implications of the functions computed by the respective gadgets.
\end{proof}

To understand how the reduction works, we take the \threesat{} instance $\psi= (X_0\lor X_1 \lor X_2) \land (\neg X_0 \lor X_1 \lor \neg X_2) \land (\neg X_1 \lor X_2 \lor X_3)$ with four propositional variables and three clauses as an example.
Furthermore, for the moment we assume that NN can take values from $\{0,1\}$ as inputs only.
Then, the resulting network $N_\psi: \Real^4 \rightarrow \Real$ computes
\begin{displaymath}
	\mathit{and}_3(\org(x_0,x_1,x_2), \allowbreak \org(\notg(x_0),x_1,\notg(x_2)), \allowbreak \org(\notg(x_1),x_2,x_3))
\end{displaymath}
and is visualized in Figure~\ref{sec:np_compl;fig:3sat_network_schema}. Note that $N_\psi$ includes redundant identity nodes matching included \notg-gadgets in order to ensure a layerwise structure.
The function computed by $N_\psi$ is described as follows. Each of the three \org-gadgets together with their connected \notg-gadgets and \id-nodes represents one
of the clauses in $\psi$. From case (\ref{prop:org_feature}) in Lemma~\ref{sec:np_completeness;lem:gadget_props} we obtain that, if an \org-gadget outputs $1$, then its current input, viewed as an assignment to the propositional variables in $\psi$, satisfies the corresponding clause. The $\mathit{and}_3$-gadget simply sums up all of its inputs and, with case (\ref{prop:andg_feature}) of Lemma~\ref{sec:np_completeness;lem:gadget_props}, we get that $y$ is equal to $3$
iff each \org-gadget outputs one. Therefore, with the output specification $\varphi_{\text{out}}(y) := y=3$ and trivial input specification $\varphi_{\text{in}}=\top$,
we get a reduction from \threesat to \reach, provided that input values are externally restricted to $\{0,1\}$.

\begin{figure}[!ht]
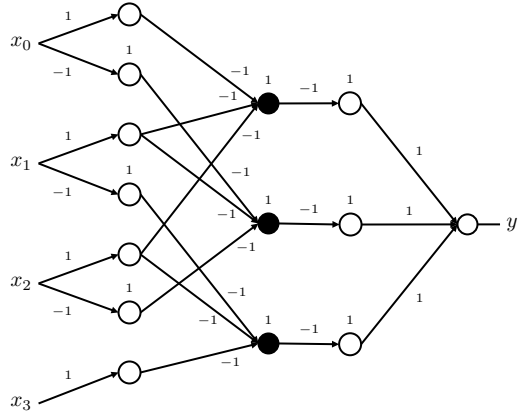

	\centering
	\include{graphics/3sat_redu_schema}\vspace*{-4mm}
	\caption{NN resulting from the reduction of the \threesat-formula $(X_0\lor X_1 \lor X_2) \land (\neg X_0 \lor X_1 \lor \neg X_2) \land (\neg X_1 \lor X_2 \lor X_3)$ under the assumption that NN are defined over $\{0,1\}$ only.}
	\label{sec:np_compl;fig:3sat_network_schema}\vspace*{-3mm}
\end{figure}

\medskip
But NN are defined for all real-valued inputs, so we need further adjustments to make the reduction complete. First, note that it
is impossible to write an input specification $\varphi_{\text{in}}(\boldsymbol{x})$ which is satisfied by $\boldsymbol{x}$ iff
$\boldsymbol{x} \in \{0,1\}^n$ because $\{0,1\}^n$ is not a hyperrectangle in $\Real^n$ and conjunctions of inequalities only specify
hyperrectangles. This is where we make use of \boolg-gadgets. Katz et al.\ claim the
following: if $x \in [0;1]$ then we have $z \in [0;\epsilon]$ iff $x \in [0;\epsilon]$ or $x \in [1-\epsilon; 1]$. Thus, by
connecting a \boolg-gadget to each input $x_i$ in $N_\psi$ including adding further identity nodes for ensuring a layerwise structure and using the simple specifications $\varphi_{\text{in}} := \bigwedge_{i=0}^3 x_i \geq 0 \land x_i \leq 1$ and
$\varphi_{\text{out}} := \bigwedge_{i=0}^3 z_i \geq 0 \land z_i \leq \epsilon \land y \geq 3(1-\epsilon) \land y \leq 3$
we would get a correct translation of $\psi$. Note that the constraint on $y$ is no longer $y=3$ as
the valid inputs to $N_\psi$, determined by the \boolg-gadgets and their output constraints, are not exactly $0$ or exactly $1$.

However, the claim about \boolg-gadgets is wrong. Consider a \boolg-gadget with very small $\epsilon$ such that it is safe to assume
$\epsilon < 2\epsilon < 1-\epsilon$. Then, for $x=2\epsilon$ we have $z = 0$, which contradicts the claim. In fact, it can be shown that for each $\epsilon \leq \frac{1}{2}$ and each input $x \in [0;1]$ the output $z$ is an element of $[0;\epsilon]$.
Clearly, this is not the intended property of these gadgets. But with some adjustments to the \boolg-gadgets we can make the reduction work.
\begin{lemma}
	Let $\boolrepairedg(x) = \id\left(\ReLU\left(\frac{1}{2} - x \right) + \ReLU\left(x - \frac{1}{2}\right) - \frac{1}{2}\right)$ be a gadget.
	It holds that $\boolrepairedg(x)=0$ if and only if $x=0$ or $x=1$.
	\label{sec:np_compl;lem:bool_gadget_prop}
\end{lemma}
\begin{proof}
	Note that the function $\boolrepairedg(x)$ is
	equivalent to
	\begin{displaymath}
		\boolrepairedg(x) = \begin{cases}
			-x & \text{if } x < \frac{1}{2} \\
			x-1 & \text{otherwise.}
		\end{cases}
	\end{displaymath}
	From this we immediately get that $\boolrepairedg(x) = 0$ if $x = 0$ or $x = 1$, and $\boolrepairedg(x) \neq 0$ for all other values of $x$.
\end{proof}

Remember that $\id$ and $\ReLU$ are PWL functions. Now, replacing all \boolg-gadgets with \boolrepairedg-gadgets in the construction and using the simple specifications $\varphi_{\text{in}} = \top$ and $\varphi_{\text{out}} = \bigwedge_{i=0}^{n-1}z_i=0 \land y=m$ for a \threesat-instance with $n$ propositional variables and $m$ clauses, we get a correct reduction from \threesat{} to \reach{}. 
\begin{theorem}
	\reach{} is NP-hard.
	\label{sec:np_compl;th:np_hardness}
\end{theorem}

One could argue that a network $N$ resulting from the reduction of \threesat{} is not a typical feed-forward neural network.
The reason is that it is common to use the same activation function across all layers , which is not the case for NN resulting from this reduction as they include both ReLU- and identity nodes in hidden layers.
However, we can modify the reduction with the help of the following lemma.
\begin{lemma}
	Let $N$ be a NN and $v$ be an identity node in $N$ computing $\sum_{i=1}^k c_i x_i + b$ and which is connected with weights $c'_1, \dotsc, c'_l$ to nodes in a following layer of $N$. We can replace $v$ with two ReLU nodes $v'$ and $v''$ such that the resulting network $N'$ computes the same function as $N$.
	\label{sec:np_compl;lem:id_replace}
\end{lemma}
\begin{proof}
	Let $N$ and $v$ be as stated above. Then, $v'$ computes $\ReLU(\sum_{i=1}^k c_i x_i + b)$ and $v''$ computes the function $\ReLU(\sum_{i=1}^k -c_i x_i - b)$.
	Furthermore, $v'$ is connected with weights $c'_1, \dotsc, c'_l$ to the same nodes as $v$ and $v''$ with weights $-c'_1, \dotsc, -c'_l$.
	It is straightforward that the resulting network $N'$ still computes the same function.
\end{proof}
We can use Lemma~\ref{sec:np_compl;lem:id_replace} to adjust the gadgets of Figure~\ref{sec:np_compl;fig:3sat_gadgets},
resulting in gadgets using ReLU-activations only, leading to an even stronger lower bound.
\begin{corollary}
	\reach{} is NP-hard if restricted to instances with networks using ReLU activations only.
	\label{cor:reluonly}
\end{corollary}


\section{NP-hardness holds in very restricted cases already}
\label{sec:low_compl_fragments}

Neural networks and specifications resulting from the reduction of \threesat{} to \reach{} presented in the previous section are already quite restricted;
the NN possess only a small, fixed number of layers and output size of one and the specifications are simple.
In this section we strengthen the NP-hardness result by constructing even simpler classes of instances for which \reach is NP-hard already.
Our main focus is on simplifying the structure of the NN rather than the specifications.
However, when trying to simplify the networks, it is tempting to increase complexity in the specifications.
To avoid this, we aim for keeping the specifications simple.

\medskip
Section~\ref{sec:singlelayer} studies the possibility to make the NN structurally as simple as possible; Sect.~\ref{sec:restricted_weights}
shows that requirements on weights and biases can be relaxed whilst retaining NP-hardness.

\subsection{Neural Networks of a simple structure}
\label{sec:singlelayer}

We consider NN with just one hidden layer of ReLU nodes and an output dimension of one. In the same way as in the previous section, we can establish a reduction from \threesat{}.
\begin{theorem}
	\reach{} is NP-hard for NN with output dimension one, a single hidden layer and simple specifications.
	\label{sec:low_compl_fagments;th:one_layer_output_one}
\end{theorem}
\begin{proof}
Let $\psi$ be a \threesat{} formula with $n$ propositional variables $X_i$ and $m$ clauses $l_j$. We slightly modify the construction
of a network $N$ in the proof of Theorem~\ref{sec:np_compl;th:np_hardness}. First, we remove the last identity node of all
\boolrepairedg-gadgets in $N$ and directly connect the two outputs of their ReLU nodes to the \andg-gadget, weighted with $1$.
Additionally, we merge \notg-gadgets and \org-gadgets in $N$. Consider an \org-gadget corresponding to some clause $l_j$. The
merged gadget has three inputs $x_{j_0}, x_{j_1}, x_{j_2} $ and computes $\ReLU\bigl(1 - \sum_{k=0}^2 f_j(x_{j_k})\bigr)$ where
$f_j(x_{j_k}) = x_{j_k}$ if $X_{j_k}$ occurs positively in $l_j$ and $f_j(x_{j_k}) = 1 - x_{ij}$ if it occurs negatively. It is
straightforward to see that the output of such a gadget is $0$ if at least one positively (resp.\ negatively) weighted input is $0$,
resp.\ $1$, and that the output is $1$ if all positively weighted inputs are $1$ and all negatively weighted inputs are $0$. These
merged gadgets are connected with weight $-1$ to the \andg-gadget. Once done for all \boolrepairedg-, \notg- and
\org-gadgets, the resulting network $N'$ computes the function
\begin{displaymath}
	\id\left(\textstyle\sum_{i=0}^{n-1} \ReLU\bigl(\frac{1}{2} - x_i\bigr) + \ReLU\bigl(x_i -\frac{1}{2}\bigr) - \textstyle\sum_{i=0}^{m-1} \ReLU\bigl(1 - \textstyle\sum_{j=0}^2 f_i(x_{ij})\bigr)\right).
\end{displaymath}
Note that $N'$ has input dimension $n$, a single hidden layer of $2n+m$ ReLU nodes and output dimension $1$ which we refer to by $y$.

\medskip
Now take the simple specifications $\varphi_{\text{in}} = \bigwedge_{i=0}^{n-1} x_i \geq 0 \land x_i \leq 1$ and
$\varphi_{\text{out}} = y = \frac{n}{2}$.  We argue that the following holds for a solution to
$(N', \varphi_{\text{in}}, \varphi_{\text{out}})$: (\textsc{i}) all $x_i$ are either $0$ or $1$, and (\textsc{ii}) the output of
each merged \org-gadget is $0$. To show (\textsc{i}), we assume the opposite, i.e.\ there is a solution with $x_{k} \in (0;1)$ for
some $k$. This implies that $\sum_{i=0}^{n-1} \ReLU\bigl(0,\frac{1}{2} - x_i\bigr) + \ReLU\bigl(0 , x_i -\frac{1}{2}\bigr) < \frac{n}{2}$
as for all $x_i \in [0;1]$ we have $\ReLU\bigl(0,\frac{1}{2} - x_i\bigr) + \ReLU\bigl(0 , x_i -\frac{1}{2}\bigr) \leq \frac{1}{2}$, and
for $x_k$ we have $\ReLU\bigl(0,\frac{1}{2} - x_k\bigr) + \ReLU\bigl(0 , x_k -\frac{1}{2}\bigr) < \frac{1}{2}$. Furthermore,
we must have $-\sum_{i=0}^{m-1} \ReLU\bigl(0, 1 - \sum_{j=0}^2 f(x_{ij})\bigr) \leq 0$. Therefore, this cannot be a solution for
$(N', \varphi_{\text{in}}, \varphi_{\text{out}})$ as it does not satisfy $y = \frac{n}{2}$.

To show (\textsc{ii}), assume there is a solution such that one merged \org-gadget outputs a value different from $0$. Then,
$-\sum_{i=0}^{m-1} \ReLU\bigl(0, 1 - \sum_{j=0}^2 f(x_{ij})\bigr) < 0$ which in combination with (\textsc{i}) yields
$y < \frac{n}{2}$. Again, this is a contradiction.

Putting (\textsc{i}) and (\textsc{ii}) together, a solution for $(N', \varphi_{\text{in}}, \varphi_{\text{out}})$ implies the existence
of a model for $\psi$. For the opposite direction assume that $\psi$ has a model $I$. Then, a solution for
$(N', \varphi_{\text{in}}, \varphi_{\text{out}})$ is given by $x_i = 1$ if $I(X_i)$ is true and $x_i = 0$ otherwise.
\end{proof}
We can use Lemma~\ref{sec:np_compl;lem:id_replace} to argue that the same holds for pure ReLU-networks. The idea here is that
replacing identity- with ReLU nodes does increase the layer sizes of a NN but not the depth of it.

In the previous section, especially in the arguments of Corollary~\ref{sec:np_compl;cor:relu_nodes}, we pointed out that the occurrence of ReLU nodes is crucial for the NP-hardness of \reach{}. Thus, it is tempting to assume that any major restriction to these nodes leads to efficiently solvable classes.
\begin{theorem}
	\reach{} is NP-hard for NN where all ReLU nodes have at most one non-zero weighted input and simple specifications.
	\label{sec:low_compl_fragments;th:relu_in_one}
\end{theorem}
\begin{proof}
	We prove NP-hardness via a reduction from \threesat. Let $\psi$ be a \threesat{} formula with $n$ propositional variables $X_i$ and $m$ clauses $l_j$.
	The reduction works in the same way as in the proof of Theorem~\ref{sec:np_compl;th:np_hardness}, but we use plain identity-nodes instead of \org-gadgets and
	we do not include an \andg-gadget. The resulting network $N_\psi: \Real^n \rightarrow \Real^{m+1}$ computes the function
	\begin{displaymath}
		\left( \id\left(\textstyle\sum_{i=1}^3 f_1(x_{1i})\right), \dotsc, \id\left(\textstyle\sum_{i=1}^3 f_m(x_{mi})\right), \id\left(\textstyle\sum_{i=1}^n \boolrepairedg(x_i) \right) \right),
	\end{displaymath}
	where $f_j(x)$ is either $\id(x)$ or $\notg(x)$, determined by the $j$-th clause of $\psi$. Furthermore, the input specification remains $\varphi_{\text{in}} = \top$ and we set the output specification to $\varphi_{\text{out}}= \bigwedge_{j=1}^m y_j \geq 1 \land z=0$,
	where $y_j$ is the output of the $j$-th identity-node corresponding to the $j$-th clause of $\psi$ and $z$ is the output corresponding to the sum of all \boolrepairedg-gadgets.
	First, note that the only ReLU nodes occurring in $N_\psi$ are inside the \boolrepairedg-gadgets, which according to Lemma~\ref{sec:np_compl;lem:bool_gadget_prop} have only one non-zero input. Second, note that the specifications $\varphi_{\text{in}}$ and $\varphi_{\text{out}}$ are simple.
	Now, if $z = 0$ then the value of an output $y_j$ is equivalent to the number of inputs equal to $1$. The remaining argument for the correctness of this reduction is the same as in the original NP-hardness proof.
\end{proof}

This result does not hold for the case of pure ReLU networks as the resulting networks include identity-nodes.
Additionally, we cannot make use of Lemma~\ref{sec:np_compl;lem:id_replace} here as it would result in NN including ReLU nodes with at least three inputs.
However, we can derive a slightly weaker result for the pure ReLU case by modifying the reduction again.
\begin{theorem}
	\reach{} is NP-hard for NN without identity-nodes and where all nodes have at most two non-zero weighted inputs and simple specifications.
\end{theorem}
\begin{proof}
    The proof works similar to the original NP-hardness argument but we use adjusted \org-gadgets.
	Such a modified gadget $\mathit{or}': \Real^3 \rightarrow \Real$ computes the function
	$\mathit{or}'(x_1,x_2,x_3) = \ReLU(1-\ReLU(\ReLU(1-x_1-x_2)-x_3))$.
	It is not hard to see that an $\mathit{or}'$-gadget fullfills the same property as an \org-gadget, namely that
	for $x_i \in \{0,1\}$ its output is $1$ if at least one $x_i$ is $1$ otherwise its $0$. Furthermore, we can see
	that this function can be represented by a gadget where all ReLU nodes have two non-zero inputs only.
	The remaining identity-nodes in other gadgets can safely be eliminated using Lemma~\ref{sec:np_compl;lem:id_replace}.
\end{proof}


\subsection{Neural Networks with simple parameters}
\label{sec:restricted_weights}

One could argue that the NP-hardness results of Sections~\ref{sec:np_compl} and \ref{sec:singlelayer} are only of limited relevance as the
constructed NN use very specific combinations of weights and biases, namely $-1, 0, \frac{1}{2}$ and $1$, which may be unlikely to
occur in this exact combination in real-world applications. To conquer such concerns, we show that \reach{} is already NP-hard in cases where only very weak
assumptions are made on the set of occurring weights and biases.

For $P \subseteq \Rat$ let $\mathop{\mathit{NN}}(P)$ be the class of NN which only use weights and biases from $P$.
Our main goal is to show that NP-hardness already occurs when $P$ contains three values: $0$, some positive and some negative value.
We make use of the same techniques as in Section~\ref{sec:np_compl} and start with defining further but similar gadgets.

\begin{definition}
	Let $c,d \in \Rat^{>0}$. We call NN, which compute one of the following functions, \emph{gadgets}:
	\begin{align*}
		\discreteg(x) &= \id\left(d-c\ReLU\left(-cx\right)-c\ReLU\left(dx\right)\right) \\
		\inverseeqg(x_1, x_2) &= \id(-cx_1 - cx_2) \\
		\normg(x) &= \id\left(-c\id\left(d-c\ReLU(-cx)\right)\right) \\
		\normnotg(x) &= \id\left(-c\ReLU\left(-c\id(-cx)\right)\right) \\
		\orleoneg(x_1,x_2,x_3) &= \id\bigl([d\cdot c^4]-c\ReLU\bigl([d\cdot c^4]+\textstyle\sum_{i=1}^3-cx_i\bigr)\bigr) \\
		\orgeqoneg(x_1,x_2,x_3) &= \id\bigl([d\cdot c^4]-c\ReLU\bigl([d\cdot c^2]+\textstyle\sum_{i=1}^3-cx_i\bigr)\bigr) \\
		\andg(x_1, \dotsc, x_n) &= \id(dx_1 + \dotsb + dx_n)
	\end{align*}
	where $[d \cdot c^n]$ with even $n$ is an abbreviation for $\underbrace{-c\cdot\id(-c\cdot\id(\dotsb(-c\cdot\id}_n(d))\dotsb))$.
	\label{sec:restricted_weights;def:gadgets}
\end{definition}
The natural candidates for these gadgets are visualized in Figure~\ref{sec:restricted_weights;fig:gadgets}.
Similar to the gadgets defined in Section~\ref{sec:np_compl} each of the above ones fullfills a specific purpose.

\begin{figure}[!h]
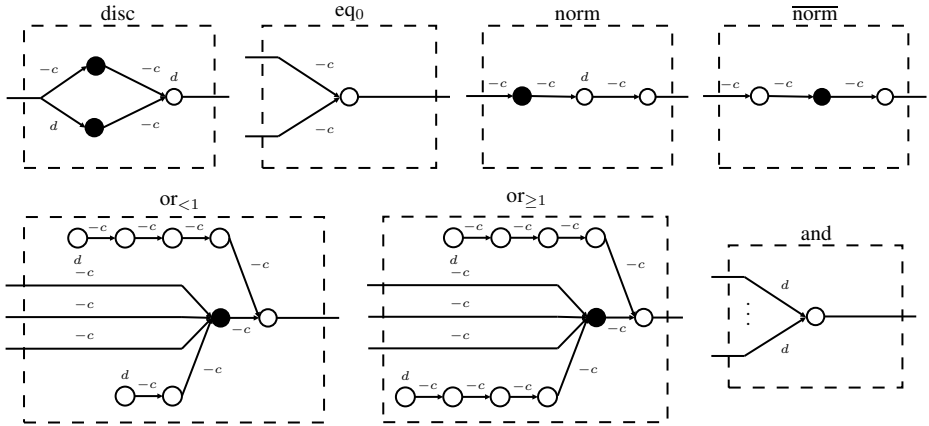

\vspace*{2mm}
	\centering
	\include{graphics/gadgets_restricted_weights}\vspace*{-4mm}
	\caption{Gadgets used to show that \reach{} is NP-hard if restricted to NN from $\mathop{\mathit{NN}}\left(\{-c,0,d\}\right)$. A non-weighted outgoing edge of a gadget is connected to a weighted incoming one of another gadget in the actual construction or are considered as outputs of the overall neural networks.}	\label{sec:restricted_weights;fig:gadgets}\vspace*{-3mm}
\end{figure}

\begin{lemma}
	Let $c,d \in \Rat^{>0}$. The following statements hold:
	\begin{enumerate}
\itemsep=0.9pt
		\item \label{prop:discreteg} $\discreteg(r) = 0$ if and only if $r = -\frac{d}{c^2}$ or $r = \frac{1}{c}.$
		\item \label{prop:eqg} $\inverseeqg(r_1,r_2) = 0$ if and only if $r_1 = -r_2$.
		\item \label{prop:normg} $\normg(-\frac{d}{c^2})=0$ and $\normg(\frac{1}{c})=-dc$.
		\item \label{prop:normnotg} $\normnotg(\frac{d}{c^2})=-dc$ and $\normnotg(-\frac{1}{c})=0$.
		\item \label{prop:orzero}$\orleoneg(0,0,0)=dc^4 - dc^5$ and $\orgeqoneg(0,0,0)=dc^4 - dc^3$
		\item \label{prop:orgeqoneg} If all $r_i \in \{-dc,0\}$ and $c<1$ then $\orleoneg(r_1,r_2,r_3)=dc^4$ if at least one $r_i=-dc$.
		\item \label{prop:orleoneg}  If all $r_i \in \{-dc,0\}$ and $c\geq 1$ then $\orgeqoneg(r_1,r_2,r_3)=dc^4$
                 if at least one $r_i=-dc$.\vspace*{-1mm}
	\end{enumerate}
	\label{sec:restricted_weights;lem:gadgets}
\end{lemma}
\begin{proof}
	We start with 
	(\ref{prop:discreteg}) and make a case distinction. If $r \leq 0$ then $\discreteg(r)=d + c^2r$.
	It is easy to see that this becomes zero if and only if $r = -\frac{d}{c^2}$.
	The next case is $r=0$ which leads to $\discreteg(0) = d$. And for the last case $r > 0 $ we get that $\discreteg(r)=d - cdr$.
	Again, it is easy to see that this becomes zero if and only if $r = \frac{1}{c}$.
	
	The remaining properties are straightforward implications of the function computed by the respective gadgets.
\end{proof}
From cases (\ref{prop:orgeqoneg}) and (\ref{prop:orleoneg}) of Lemma~\ref{sec:restricted_weights;lem:gadgets} we can see that the properties of
$\orgeqoneg$ and $\orleoneg$ are the same,
the only difference is that $\orgeqoneg$ fullfills it for the case that $c\geq 1$ and $\orleoneg$ fullfills it in the $c < 1$ case.

Let $\psi$ be some \threesat-formula. We construct a network $N_{\psi,c,d}$, using the gadgets above, which is later used in a reduction from \threesat
to \reach restricted to NN using only $-c,0$ and $d$ as weights or bias.
\begin{definition}
Let $\psi$ be a \threesat-formula with $n$ propositional variables $X_i$ and $m$ clauses, $c,d \in \Rat^{>0}$ and $i \in \{1,\dotsc,n\}$. The network
$N_{\psi,c,d}$ computes a function $\Real^{2n} \rightarrow \Real^{2n+1}$ with two input dimensions $x_i$ and $\overline{x_i}$ per propositional variable $X_i$.
We describe the structure of $N_{\psi,c,d}$ by defining the function computed for each output dimension:
\begin{align*}
	z_i = &\id(-c\id(-c\id(-c\id(-c\id(-c\discreteg(x_i)))))) \\
	e_i = &\id(-c\id(-c\id(-c\id(-c\id(-c\id(-c\inverseeqg(x_i,\overline{x_i}))))))) \\
	y = &\begin{cases}
	\mathop{\mathit{and}_m}\bigl(\orgeqoneg(\id(f_{i^1_1}),\id(f_{i^1_2}),\id(f_{i^1_3}), \dotsc, \orgeqoneg(\id(f_{i^m_1}),\id(f_{i^m_2}),\id(f_{i^m_3}))\bigr)
	&, \text{ if } c \geq 1 \\
	\mathop{\mathit{and}_m}\bigl(\orleoneg(\id(f_{i^1_1}),\id(f_{i^1_2}),\id(f_{i^1_3}), \dotsc, \orleoneg(\id(f_{i^m_1}),\id(f_{i^m_2}),\id(f_{i^m_3}))\bigr)
	&, \text{ otherwise,}
	\end{cases}
\end{align*}
where $f_{i^j_k} = \normg(x_i)$ if the variable $X_i$ is at the $k$-th position in the $j$-th clause and occurs positively and
$f_{i^j_k} = \normnotg(\overline{x_i})$ if it occurs negatively.	
\end{definition}
Note that the redundant identity nodes ensure a layerwise structure and that each such network $N_{\psi,c,d}$ has eight layers.
Moreover, from the definition of the gadgets and the construction above follows that $N_{\psi,c,d} \in \mathit{NN}(\{-c,0,d\})$.

\medskip
With this definition at hand, we are suited to prove our main statement of this section. To keep the proof clear and w.l.o.g., we do not distinguish
$\orgeqoneg$- and $\orleoneg$-gadgets and simply call them $\org$-gadgets.
\begin{theorem}
	Let $c,d \in \Rat^{>0}$. \reach{} restricted to neural networks from $\mathit{NN}(\{-c,0,d\})$ and simple specifications is NP-hard.
	\label{sec:restricted_weights;th:np_compl}
\end{theorem}
\begin{proof}
	We prove the result above via reduction from \threesat.
	Take a \threesat-formula $\psi$ with $n$ propositional variables $X_i$ and $m$ clauses $l_j$ and consider an \reach-instance
	$(N_{\psi,c,d}, \varphi_{\text{in}}, \varphi_{\text{out}})$ with $N_{\psi,c,d}$ defined above, $\varphi_{\text{in}}=\top$ and
	$\varphi_{\text{out}}=\bigwedge_{i=1}^{n} z_i = 0 \land e_i = 0 \land y = m\cdot d^2c^4$. Obviously, these specifications are simple and
	clearly $(N_{\psi,c,d}, \varphi_{\text{in}}, \varphi_{\text{out}})$ can be constructed in polynomial time in the size of $\psi$.
	
\medskip
	For the correctness of the construction we start with assuming that $\psi$ has a model $I$. We claim that a solution for
	$(N_{\psi,c,d}, \varphi_{\text{in}}, \varphi_{\text{out}})$ is given by $x_i = \frac{1}{c}$ if $I(X_i)$ is true,
	$x_i = -\frac{d}{c^2}$ otherwise, and $\overline{x_i}=-x_i$. Note that $\varphi_{\text{in}}$ is trivially satisfied by this assignment.
	So apply these inputs to $N_{\psi,c,d}$. According to cases (\ref{prop:discreteg}) and (\ref{prop:eqg}) of
	Lemma~\ref{sec:restricted_weights;lem:gadgets}, all outputs $z_i$ and $e_i$ are $0$.
	Thus, it is left to argue that $y = m\cdot d^2c^4$. Consider one of the $\org$-gadgets occurring in $N_{\psi,c,d}$, corresponding to a clause $l_j$.
	Its inputs are given by the \normg- and \normnotg-gadgets connected to the inputs $x_i$, resp.\ $\overline{x_i}$ corresponding to the $X_i$
	occurring in $l_j$. According to cases (\ref{prop:normg}) and (\ref{prop:normnotg}) of Lemma~\ref{sec:restricted_weights;lem:gadgets},
	these inputs are either $0$ or $-dc$.
	If $l_j$ is satisfied by $I$ then there is at least one input to the $\org$-gadget that is equal to $-dc$.
	From the fact that $\psi$ is satisfied by $I$ and the cases (\ref{prop:orgeqoneg}) and (\ref{prop:orleoneg}) of
	Lemma~\ref{sec:restricted_weights;lem:gadgets}
	we get that each $\org$-gadget outputs $dc^4$.
	Therefore, the output $y$ of $N_{\psi,c,d}$ is $m \cdot d^2c^4$. This means that $\varphi_{\text{out}}$ is valid as well.
	
	Consider now the converse direction. A solution for $(N_{\psi,c,d}, \varphi_{\text{in}}, \varphi_{\text{out}})$ must yield that all $x_i$ are $\frac{1}{c}$ or $-\frac{d}{c^2}$ and $x_i = \overline{x_i}$ as all $z_i$ and $e_i$ have to equal $0$.
	Therefore the output of each $\normg$- and $\normnotg$ in $N_{\psi,c,d}$ is either $-dc$ or $0$.
    This implies that all $m$ $\org$-gadgets have to output $dc^4$ as $y$ must equal $m\cdot d^2c^4$.
    According to cases (\ref{prop:orzero}), (\ref{prop:orgeqoneg}) and (\ref{prop:orleoneg}) of Lemma~\ref{sec:restricted_weights;lem:gadgets},
    each $\org$-gadget has at least one input that is $-dc$ which in turn means that there is at least one indirectly connected $x_i$ or $\overline{x_i}$ that is $\frac{1}{c}$ resp. $\frac{d}{c^2}$.
    Thus, $\psi$ is satisfied by setting $X_i$ true if $x_i = \frac{1}{c}$ and false if $x_i = -\frac{d}{c^2}$.
\end{proof}

If $d=c$ and we allow for arbitrary specifications, we can show that $0$ as a value for weights or biases is not required to keep the lower bound of
\reach{}.
\begin{theorem}
	Let $c \in \Rat^{>0}$. \reach is NP-hard restricted to neural networks from $\mathit{NN}(\{-c,c\})$ and arbitrary specifications.
	\label{sec:restricted_weights;th:no_zero}
\end{theorem}

\begin{proof}
	This is done in the same way as the proof of Theorem~\ref{sec:restricted_weights;th:np_compl} with some slight modifications to the resulting NN and specifications.
	We only sketch this reduction by describing the differences compared to the instances $(N_{\psi,c,c}, \varphi_{\text{in}}, \varphi_{\text{out}})$ resulting from the reduction used in Theorem~\ref{sec:restricted_weights;th:np_compl}.
	Our goal is to eliminate all zero weights and bias in $N_{\psi,c,c}$ while ensuring that the modified network $N_{\psi,c,c}' \in \mathop{\mathit{NN}}(\{-c,c\})$ has the same functional properties.
	Therefore, we split the argument into two parts: first, we argue how to eliminate zero weights, and second,
	we argue how to eliminate zero bias.

\medskip	
	To get rid of zero weights, we add for each $x_i$ the conjunct $x_i = -\overline{x_i}$ to the input specification.
	As a side effect, this makes the usage of \inverseeqg-gadgets obsolete, though we do not include them and the following chain of $\id$-nodes
	and we do not include $\bigwedge_{i=0}^{n-1} e_i = 0$ in the output specification as well.
	Now, consider the weights between the input and the first hidden layer of $N_{\psi,c,c}'$.
	If the inputs $x_i$ and $\overline{x_i}$ were originally weighted with zero in some node in the first hidden layer of $N_{\psi,c,c}$,
	we set the weights corresponding to $x_i$ and $\overline{x_i}$ to be $c$ in $N_{\psi,c,c}'$.
	In combination with the input constraint $x_i = -\overline{x_i}$ this is equal to weighting $x_i$ and $\overline{x_i}$ with zero.
	If $x_i$ ($\overline{x_i}$) was weighted with $c$ in $N_{\psi,c,c}$, we set the weight of $\overline{x_i}$ ($x_i$) to be $-c$ and
	if it was weighted with $-c$, we set the weight of its counterpart to be $c$ in $N_{\psi,c,c}'$.
	This leads to the case that all non-zero inputs of a node in the first hidden layer are doubled compared to the same inputs in $N_{\psi,c,c}$.
	Consider now the weights between two layers $l$ and $l+1$ with $l > 0$ of $N_{\psi,c,c}$.
	For each node in $l$ we add an additional node in the layer $l$ of $N_{\psi,c,c}'$ with the same input weights.
	If the output of a node in layer $l$ was originally weighted with zero in $N_{\psi,c,c}$
	then we weight it with $c$ and the corresponding output of its copy with $-c$ in $N_{\psi,c,c}'$.
	If the output was originally weighted with weight $c$ ($-c$) in $N_{\psi,c,c}$
	then we weight the output of the copy node with $c$ ($-c$) $N_{\psi,c,c}'$, too.
	As before, this doubles the input values at the nodes in layer $l+1$ of $N_{\psi,c,c}'$ compared to $N_{\psi,c,c}$,
	which means that in comparison to $N_{\psi,c,c}$ the output value of $N_{\psi,c,c}'$ is multiplied by $2^7$.
	Thus, we have to change the output constraint of $y$ to be $y = 2^7(m \cdot c^6)$.
	
	To get rid of zero bias, we add the inputs $x_{\text{bias},1}, \overline{x_{\text{bias},1}}, \dotsc , x_{\text{bias},7},\overline{x_{\text{bias},7}}$ to $N_{\psi,c,c}'$ and the constraints $x_{\text{bias},i} = -\sum_{j=0}^{i-1}\frac{1}{2^{j+1}c^j}$ and $x_{\text{bias},i} = - \overline{x_{\text{bias},i}}$ to the input specification.
	Next, we set the bias of all nodes which originally had a zero bias to be $c$.
	For $x_{\text{bias},i}$ with $i > 1$ we add a chain of $i-1$ identity nodes each with bias $c$ and interconnected with weight $c$
	and connect this chain with weight $c$ to $x_{\text{bias},i}$ and $-c$ to $\overline{x_{\text{bias},i}}$.
	All other incoming weights of this chain are assumed to be zero which is realized using the same techniques as described in the previous paragraph.
	Now, if a node in the first hidden layer of $N_{\psi,c,c}$ has a zero bias in $N_{\psi,c,c}$,
	we give weight $c$ to the input $x_{\text{bias},1}$ and $-c$ to $\overline{x_{\text{bias},1}}$ in $N_{\psi,c,c}'$.
	If the input specification holds then the bias of the node plus these inputs sum up to zero.
	If a node in some layer $l \in\{2, \dotsc, 7\}$ has a zero bias in $N_{\psi,c,c}$,
	we give weight $c$ to the output of the last node of the chain corresponding to $x_{\text{bias},l}$ and its copy, resulting from the zero weight
	elimination of the previous paragraph.
	Again, if the input specification holds, the bias value of this node is nullified.

\medskip	
	In summary, we get the specifications $\varphi_{\text{in}}'= \bigwedge_{i=1}^n x_i = \overline{x_i} \wedge \bigwedge_{i=1}^7 x_{\text{bias},i} = -\sum_{j=0}^{i-1}\frac{1}{2^{j+1}c^j} \wedge x_{\text{bias},i} = - \overline{x_{\text{bias},i}}$
	and $\varphi_{\text{out}}' = \bigwedge_{i=1}^n z_i = 0 \wedge y=2^7(m \cdot c^6)$. Note that $\varphi_{\text{in}}'$ is no longer simple and that
	$N_{\psi,c,c}' \in \mathop{\mathit{NN}}(\{-c,c\})$.
\end{proof}


\section{Conclusion and outlook}
\label{sec:conclusion}
We investigated the computational complexity of the reachability problem for NN with PWL activations, including the common activation function ReLU, and input/output specifications
given by conjunctions of linear inequalities over the respective input/output dimensions.
We revised the original proof of its NP-completeness, fixing flaws in both the upper and lower bound, and showed that the parameter driving NP-hardness is
the number of PWL nodes. Furthermore, we showed that \reach is difficult for very restricted classes of small neural networks already, respectively two parameters
of different signum and $0$ occurring as weights and biases suffice for NP-hardness. This indicates that finding non-trivial classes of NN and specifications with
practical relevance and polynomial \reach{} is unlikely.

Additionally, modern uses of deep learning techniques limit the practical applicability of findings about \reach even more.
There are two reasons for this. On the one hand, the deep learning framework circumvents far more techniques and models than our considered NN model.
On the other hand, usual deep learning applications are concerned with data of very high complexity and dimensionality so that simple linear specifications
as they are considered here are hardly expressive enough for formalising meaningful properties of the input data and possibly also the output values.

The former concern is tackled by extending the analysis to similar reachability problems of different models, like Convolutional Neural Networks (CNN) \cite{KhanSZQ20} or Graph Neural Networks (GNN) \cite{WuPCLZY21}.
The former incorporates different kinds of layers like convolutional layers or pooling layers. However, in the classical sense, these are just layers consisting of nodes with PWL activations
like ReLU, which makes it straightforward to argue that reachability for CNN is NP-complete, too.
The case of GNN is more challenging as these are neural network based models which compute functions over graphs. First results on this matter are provided in \cite{SL23},
showing that there are computational limits for the possibility of formal verifying reachability properties of common GNN models.

The latter concern is best understood by example. Consider the problem \reach{} in the context of an image classification task like classifying whether an
x-ray picture shows a broken or healthy bone. It immediately becomes clear that our current definition of specifications is not suited to
express properties like ``the picture shows a bone with a hairline fracture''. But expressing such properties is obviously desired for tasks like this,
in order to rule out misclassification of severely damaged bones.
Thus, we need specifications of higher expressibility which directly leads to the question whether the upper complexity bound, namely the membership to NP,
holds for the reachability problem with specifications extended in a particular way.
Starting from our current definition of specifications, a natural extension is allowing for full boolean logic with linear inequalities as propositions,
which means adding disjunction and negation operators to our current definition. It is not hard to see that adding disjunction does not change the upper bound.
The idea is that we guess for each occurring disjunction which disjunct is satisfied and include this guess in the witness for the NP membership.
Obviously, this part of the witness is polynomially bounded by the number of disjunctions in the specifications and therefore by the size of the specifications.
As soon as we allow for disjunctions, adding the negation operator does not harm the upper bound either.
We can bring each such specification into a normal form  (for example conjunctive normal form) where all negations are directly in front of propositions like $\neg (x \leq 2)$ which we interpret as $x > 2$.
We provided an argument how to handle strict inequalities in the proof of Lemma~\ref{thm:relulinprognp}.
In summary, full boolean expessibility in the specifications does not change the upper bound for the reachability problem.
The obvious next step is adding existential and universal quantifiers to the specifications.
Allowing for pure existential quantification in the specifications just leads to an equivalent reformulation of the reachability problem itself. Currently,
\reach{} is formulated such that it existentially quantifies over the free variables in the input and output specifications. Using existential quantifiers,
a possible reformulation leads to a normal form of specifications like $\varphi = \exists x_1. \exists x_2. x_1 - x_2 \geq 3 \land N(x_1,x_2) \geq 0$ where $N$ is the
considered NN and the question is whether $\varphi$ is satisfiable. It is not hard to see that such a reformulation does not harm the membership in NP. Such specifications
with pure universal quantification are not more difficult then, either, because of $\forall x. \psi \equiv \neg \exists x. \neg\psi$. Hence, we simply get that
reachability with purely universal specifications is Co-NP-complete. It remains to be seen whether specifications of certain quantifier alternation depth will yield
NN-reachability problems that are complete for the corresponding classes in the polynomial hierarchy \cite{Stockmeyer:1976:PH}.

\bibliographystyle{fundam}
\bibliography{ref}

\begin{thebibliography}{10}
\providecommand{\url}[1]{\texttt{#1}}
\providecommand{\urlprefix}{URL }
\expandafter\ifx\csname urlstyle\endcsname\relax
  \providecommand{\doi}[1]{doi:\discretionary{}{}{}#1}\else
  \providecommand{\doi}{doi:\discretionary{}{}{}\begingroup
  \urlstyle{rm}\Url}\fi
\providecommand{\eprint}[2][]{\url{#2}}

\bibitem{KrizhevskySH17}
Krizhevsky A, Sutskever I, Hinton GE.
\newblock ImageNet classification with deep convolutional neural networks.
\newblock \emph{Commun. {ACM}}, 2017.
\newblock \textbf{60}(6):84--90.
\newblock \doi{10.1145/3065386}.

\bibitem{X12a}
Hinton G, Deng L, Yu D, Dahl GE, Mohamed Ar, Jaitly N, Senior A, Vanhoucke V,
  Nguyen P, Sainath TN, et~al.
\newblock Deep Neural Networks for Acoustic Modeling in Speech Recognition: The
  Shared Views of Four Research Groups.
\newblock \emph{{IEEE} Signal Process. Mag.}, 2012.
\newblock \textbf{29}(6):82--97.
\newblock \doi{10.1109/MSP.2012.2205597}.

\bibitem{GrigorescuTCM20}
Grigorescu SM, Trasnea B, Cocias TT, Macesanu G.
\newblock A survey of deep learning techniques for autonomous driving.
\newblock \emph{J. Field Robotics}, 2020.
\newblock \textbf{37}(3):362--386.
\newblock \doi{10.1002/rob.21918}.

\bibitem{LitjensKBSCGLGS17}
Litjens G, Kooi T, Bejnordi BE, Setio AAA, Ciompi F, Ghafoorian M, van~der Laak
  JAWM, van Ginneken B, S{\'{a}}nchez CI.
\newblock A survey on deep learning in medical image analysis.
\newblock \emph{Medical Image Anal.}, 2017.
\newblock \textbf{42}:60--88.
\newblock \doi{10.1016/j.media.2017.07.005}.

\bibitem{DixonKB17}
Dixon M, Klabjan D, Bang JH.
\newblock Classification-based financial markets prediction using deep neural
  networks.
\newblock \emph{Algorithmic Finance}, 2017.
\newblock \textbf{6}(3-4):67--77.
\newblock \doi{10.3233/AF-170176}.

\bibitem{HuangKRSSTWY20}
Huang X, Kroening D, Ruan W, Sharp J, Sun Y, Thamo E, Wu M, Yi X.
\newblock A survey of safety and trustworthiness of deep neural networks:
  Verification, testing, adversarial attack and defence, and interpretability.
\newblock \emph{Comput. Sci. Rev.}, 2020.
\newblock \textbf{37}:100270.
\newblock \doi{10.1016/j.cosrev.2020.100270}.

\bibitem{KatzBDJK17}
Katz G, Barrett CW, Dill DL, Julian K, Kochenderfer MJ.
\newblock Reluplex: An Efficient {SMT} Solver for Verifying Deep Neural
  Networks.
\newblock In: Majumdar R, Kuncak V (eds.), Computer Aided Verification - 29th
  International Conference, {CAV} 2017, Heidelberg, Germany, July 24-28, 2017,
  Proceedings, Part {I}, volume 10426 of \emph{Lecture Notes in Computer
  Science}. Springer, 2017 pp. 97--117.
\newblock \doi{10.1007/978-3-319-63387-9\_5}.

\bibitem{Ehlers17}
Ehlers R.
\newblock Formal Verification of Piece-Wise Linear Feed-Forward Neural
  Networks.
\newblock In: D'Souza D, Kumar KN (eds.), Automated Technology for Verification
  and Analysis - 15th International Symposium, {ATVA} 2017, Pune, India,
  October 3-6, 2017, Proceedings, volume 10482 of \emph{Lecture Notes in
  Computer Science}. Springer, 2017 pp. 269--286.
\newblock \doi{10.1007/978-3-319-68167-2\_19}.

\bibitem{NarodytskaKRSW18}
Narodytska N, Kasiviswanathan SP, Ryzhyk L, Sagiv M, Walsh T.
\newblock Verifying Properties of Binarized Deep Neural Networks.
\newblock In: McIlraith SA, Weinberger KQ (eds.), Proceedings of the
  Thirty-Second {AAAI} Conference on Artificial Intelligence, (AAAI-18), the
  30th innovative Applications of Artificial Intelligence (IAAI-18), and the
  8th {AAAI} Symposium on Educational Advances in Artificial Intelligence
  (EAAI-18), New Orleans, Louisiana, USA, February 2-7, 2018. {AAAI} Press,
  2018 pp. 6615--6624.
\newblock
  \urlprefix\url{https://www.aaai.org/ocs/index.php/AAAI/AAAI18/paper/view/16898}.

\bibitem{BunelTTKM18}
Bunel R, Turkaslan I, Torr PHS, Kohli P, Mudigonda PK.
\newblock A Unified View of Piecewise Linear Neural Network Verification.
\newblock In: Bengio S, Wallach HM, Larochelle H, Grauman K, Cesa{-}Bianchi N,
  Garnett R (eds.), Advances in Neural Information Processing Systems 31:
  Annual Conference on Neural Information Processing Systems 2018, NeurIPS
  2018, December 3-8, 2018, Montr{\'{e}}al, Canada. 2018 pp. 4795--4804.
\newblock
  \urlprefix\url{https://proceedings.neurips.cc/paper/2018/hash/be53d253d6bc3258a8160556dda\linebreak3e9b2-Abstract.html}.

\bibitem{RuanHK18}
Ruan W, Huang X, Kwiatkowska M.
\newblock Reachability Analysis of Deep Neural Networks with Provable
  Guarantees.
\newblock In: Lang J (ed.), Proceedings of the Twenty-Seventh International
  Joint Conference on Artificial Intelligence, {IJCAI} 2018, July 13-19, 2018,
  Stockholm, Sweden. ijcai.org, 2018 pp. 2651--2659.
\newblock \doi{10.24963/ijcai.2018/368}.

\bibitem{RuanHK18_arxiv_version}
Ruan W, Huang X, Kwiatkowska M.
\newblock Reachability Analysis of Deep Neural Networks with Provable
  Guarantees.
\newblock \emph{CoRR}, 2018.
\newblock \textbf{abs/1805.02242}.
\newblock \eprint{1805.02242}, \urlprefix\url{http://arxiv.org/abs/1805.02242}.

\bibitem{Katz21}
Katz G, Barrett CW, Dill DL, Julian K, Kochenderfer MJ.
\newblock Reluplex: a calculus for reasoning about deep neural networks.
\newblock \emph{Form Methods Syst Des}, 2021.
\newblock \doi{10.1007/s10703-021-00363-7}.

\bibitem{SalzerL21}
S{\"{a}}lzer M, Lange M.
\newblock Reachability is NP-Complete Even for the Simplest Neural Networks.
\newblock In: Reachability Problems - 15th International Conference, {RP} 2021,
  Liverpool, UK, October 25-27, 2021, Proceedings, volume 13035 of
  \emph{Lecture Notes in Computer Science}. Springer, 2021 pp. 149--164.
\newblock \doi{10.1007/978-3-030-89716-1\_10}.

\bibitem{Karmarkar84}
Karmarkar N.
\newblock A new polynomial-time algorithm for linear programming.
\newblock \emph{Comb.}, 1984.
\newblock \textbf{4}(4):373--396.
\newblock \doi{10.1007/BF02579150}.

\bibitem{AkitundeLMP18}
Akintunde M, Lomuscio A, Maganti L, Pirovano E.
\newblock Reachability Analysis for Neural Agent-Environment Systems.
\newblock In: Thielscher M, Toni F, Wolter F (eds.), Principles of Knowledge
  Representation and Reasoning: Proceedings of the Sixteenth International
  Conference, {KR} 2018, Tempe, Arizona, 30 October - 2 November 2018. {AAAI}
  Press, 2018 pp. 184--193.
\newblock \urlprefix\url{https://aaai.org/ocs/index.php/\linebreak
  KR/KR18/paper/view/17991}.

\bibitem{Karp72}
Karp RM.
\newblock Reducibility Among Combinatorial Problems.
\newblock In: Miller RE, Thatcher JW (eds.), Proceedings of a symposium on the
  Complexity of Computer Computations, held March 20-22, 1972, at the {IBM}
  Thomas J. Watson Research Center, Yorktown Heights, New York, {USA}, The
  {IBM} Research Symposia Series. Plenum Press, New York, 1972 pp. 85--103.
\newblock \doi{10.1007/978-1-4684-2001-2\_9}.

\bibitem{KorteV06}
Korte B, Vygen J.
\newblock Combinatorial Optimization.
\newblock Springer Berlin, Heidelberg, 2006.
\newblock ISBN 978-3-540-29297-5.
\newblock \doi{10.1007/3-540-29297-7}.

\bibitem{KhanSZQ20}
Khan A, Sohail A, Zahoora U, Qureshi AS.
\newblock A survey of the recent architectures of deep convolutional neural
  networks.
\newblock \emph{Artif. Intell. Rev.}, 2020.
\newblock \textbf{53}(8):5455--5516.
\newblock \doi{10.1007/s10462-020-09825-6}.

\bibitem{WuPCLZY21}
Wu Z, Pan S, Chen F, Long G, Zhang C, Yu PS.
\newblock A Comprehensive Survey on Graph Neural Networks.
\newblock \emph{{IEEE} Trans. Neural Networks Learn. Syst.}, 2021.
\newblock \textbf{32}(1):4--24.
\newblock \doi{10.1109/TNNLS.2020.2978386}.

\bibitem{SL23}
S{\"a}lzer M, Lange M.
\newblock Fundamental Limits in Formal Verification of Message-Passing Neural
  Networks.
\newblock In: The Eleventh International Conference on Learning
  Representations. 2023
  \urlprefix\url{https://openreview.net/forum?id=WlbG820mRH-}.

\bibitem{Stockmeyer:1976:PH}
Stockmeyer LJ.
\newblock The polynomial-time hierarchy.
\newblock \emph{Theor.\ Comp.\ Sci.}, 1976.
\newblock \textbf{3}(1):1--22.

\end{thebibliography}

\end{document}